\documentclass[11pt]{article}
\usepackage{amsmath, amsfonts, url}
\usepackage[T1]{fontenc}
\usepackage[latin1]{inputenc}

\usepackage{graphicx}
\usepackage{amssymb}

\usepackage{lineno}

\begin{document}

\makeatletter

\newenvironment{algorithm}{\begin{algorithm1}\ \\
    \vspace{-0.2cm}}{\end{algorithm1}}

\newenvironment{proofsk}{\begin{proof}[Proof Sketch:]}
{\end{proof}}

\newenvironment{smallproof}{\nopagebreak \begin{quote} %
\begin{small} \noindent{\bf Proof:}}{ \qed \par %
\end{small} \end{quote} \medskip}

\newenvironment{note}{\nopagebreak \begin{quote} %
\noindent{\bf Note:}}{%
\end{quote} \medskip}

\newenvironment{notes}{\nopagebreak \begin{quote} %
\noindent{\bf Notes:} \par%
\begin{itemize}}{%
\end{itemize}\end{quote} \medskip}

\newenvironment{summary}{\begin{quote} {\bf Summary:}}{\end{quote}}


\newcommand{\eqdef}{\stackrel{def}{=}}
\newcommand{\N}{\mathbb{N}}
\newcommand{\R}{\mathbb{R}}
\newcommand{\C}{\mathbb{C}}
\newcommand{\Z}{\mathbb{Z}}
\newcommand{\F}{\mathbb{F}}
\newcommand{\Zn}{{\Z}_n}
\newcommand{\bits}{\{0,1\}}
\newcommand{\inr}{\in_{\mbox{\tiny R}}}
\newcommand{\getsr}{\gets_{\mbox{\tiny R}}}
\newcommand{\st}{\mbox{ s.t. }}
\newcommand{\etal}{{\it et al }}
\newcommand{\into}{\rightarrow}
\newcommand{\ExactT}{$f^G_{k,T}$~}
\newcommand{\Exactn}{{\rm Exactly-}$n$~}
\newcommand{\ExactN}{{\rm Exactly-}$N$~}

\newcommand{\Ex}{\mathbb{E}}
\newcommand{\To}{\rightarrow}
\newcommand{\e}{\epsilon}
\newcommand{\ee}{\varepsilon}
\newcommand{\ceil}[1]{{\lceil{#1}\rceil}}
\newcommand{\floor}[1]{{\lfloor{#1}\rfloor}}
\newcommand{\angles}[1]{\langle #1 \rangle}
\newcommand{\var}{\mbox{var}}
\newcommand{\trace}{\mbox{trace}}
\newcommand{\ignore}[1]{}
\newcommand{\Alg}{\mathrm{Alg}}
\newcommand{\fro}[1]{\|#1\|_F}
\newcommand{\trn}[1]{\|#1\|_{tr}}
\newcommand{\norm}[1]{\|#1\|}

\newcommand{\NP}{\mathbf{NP}}
\renewcommand{\P}{\mathbf{P}}
\newcommand{\PCP}{\mathbf{PCP}}
\newcommand{\RP}{\mathbf{RP}}
\newcommand{\BPP}{\mathbf{BPP}}
\newcommand{\Tr}{\mathrm{Tr}}

\newcommand{\lang}{\mathcal{L}}

\newcommand{\Poincare}{Poincar�}


\newtheorem{theorem}{Theorem}
\newtheorem{lemma}[theorem]{Lemma}
\newtheorem{fact}[theorem]{Fact}
\newtheorem{claim}[theorem]{Claim}
\newtheorem{corollary}[theorem]{Corollary}
\newtheorem{conjecture}[theorem]{Conjecture}
\newtheorem{question}[theorem]{Question}
\newtheorem{proposition}[theorem]{Proposition}
\newtheorem{axiom}[theorem]{Axiom}
\newtheorem{remark}[theorem]{Remark}
\newtheorem{example}[theorem]{Example}
\newtheorem{exercise}[theorem]{Exercise}
\newtheorem{definition}[theorem]{Definition}
\newtheorem{observation}[theorem]{Observation}

\def\pproof{\par\penalty-1000\vskip .5 pt\noindent{\bf Proof\/ }}
\newcommand{\QED}{\hfill$\;\;\;\rule[0.1mm]{2mm}{2mm}$}

\newenvironment{proof}{\begin{pproof}}{\QED\end{pproof}~\\}


\newcommand{\NChooseM}[2]{\ensuremath{\lp{(}\begin{array}{cc}#1\\#2\end{array}\rp{)}}}
\newcommand{\mc}[1]{{\cal{#1}}}
\newcommand{\lp}[1]{\left #1}
\newcommand{\rp}[1]{\right #1}
\newcommand{\vect}[1]{\ensuremath{{\mathbf #1}}}
\newcommand{\IP}[2]{\ensuremath{\lp{<}#1,#2\rp{>}}}
\newcommand{\sv}{\ensuremath{\vect{v}}}
\newcommand{\sve}{\ensuremath{v}}

\newcommand{\op}{W}
\newcommand{\opr}{\vect{w}}
\newcommand{\allvects}[1]{L_{#1}}
\newcommand{\avn}[1]{L_{#1}}
\newcommand{\av}{L_n}
\newcommand{\ff}{\lfloor f \rfloor}

\newcommand{\sizeo}[1][\e]{\mathrm{size}^{(v)}_{#1}}
\newcommand{\corro}[2][\e]{\mathrm{ubc}^{(#2)}_{#1}}
\newcommand{\corr}[1][\e]{\mathrm{ubc}_{#1}}
\newcommand{\rcorro}[2][\e]{\mathrm{corr}^{(#2)}_{#1}}
\newcommand{\rcorr}[1][\e]{\mathrm{corr}_{#1}}
\newcommand{\size}[1][\e]{\mathrm{size}_{#1}}
\newcommand{\scorr}[1][\e]{\mathrm{sc}_{#1}}
\newcommand{\monoo}[2][\rho]{\mathrm{mono}^{(#2)}_{#1}}
\newcommand{\mono}[1][\rho]{\mathrm{mono}_{#1}}
\newcommand{\hmonoo}[2][\rho]{\mathrm{hmono}^{(#2)}_{#1}}
\newcommand{\hmono}[1][\rho]{\mathrm{hmono}_{#1}}

\newcommand{\rk}{\mathrm{rank}}

\newcommand{\ind}{\ensuremath{\alpha}}
\newcommand{\chr}{\ensuremath{\chi}}

\newcommand{\rs}{Ruzsa-Szemer\'{e}di}

\title{On The Communication Complexity of\\ High-Dimensional Permutations}
\author{Nati Linial\thanks{Supported in part by ERC grant 339096, High-dimensional combinatorics.}\\
	Hebrew University of Jerusalem\\
	Jerusalem, Israel \\
    	{\tt nati@cs.huji.ac.il}
\and Toniann Pitassi \\
   	University of Toronto and the Institute of Advanced Studies\\
    	Toronto, Canada and Princeton, U.S.A\\
    	{\tt toni@cs.toronto.edu}
\and Adi Shraibman \\
   	The Academic College of Tel-Aviv-Yaffo \\
    	Tel-Aviv, Israel \\
    	{\tt adish@mta.ac.il}
}

\date{}

\maketitle

\begin{abstract}
We study the multiparty communication complexity of
high dimensional permutations, in the Number On the Forehead (NOF) model. This model is due to Chandra, Furst and Lipton (CFL) who also gave a nontrivial protocol for the \Exactn problem where three players receive integer inputs and need to decide if their inputs sum to a given integer $n$. There is a considerable body of literature dealing with the same problem, where $(\mathbb{N},+)$ is replaced by some other abelian group.
Our work can be viewed as a far-reaching extension of this line of work.

We show that the known lower bounds for that group-theoretic problem apply to all high dimensional permutations. We introduce new proof techniques that appeal to recent advances in Additive Combinatorics and Ramsey theory. We reveal new and unexpected connections between the NOF communication complexity of high dimensional permutations and a variety of well known and thoroughly studied problems in combinatorics. 

Previous protocols for \Exactn all rely on the construction of large sets of integers 
without a 3-term arithmetic progression. No direct {\em algorithmic} protocol was 
previously known for the problem, and we provide the first such algorithm. This suggests new ways to significantly improve the CFL protocol. 

Many new open questions are presented throughout.
\end{abstract}

\section{Introduction}

The multiplayer Number On the Forehead (NOF) model of communication complexity
was first introduced by Chandra, Furst and Lipton \cite{CFL83}.
Here $k$ players need to evaluate a given function $f: [n]^{k}\to \{0,1\}$. We actually think of $f$ as having $k$ arguments $x_1,\ldots,x_k$, each comprised of $\log n$ bits.
The $i$-th input vector $x_i$ is placed metaphorically on player $i$'s forehead, so that every player sees the whole input but one argument. 
Players communicate by writing bits on a shared blackboard (according
to an agreed-upon protocol) in order to compute $f$.

The NOF communication model has turned out to be a fascinating, though exceedingly difficult object of study. Indeed, good lower bounds in the NOF model would resolve several longstanding open problems in complexity theory, such as lower bounds on the size of $ACC^0$ circuits for a natural function in $P$ \cite{Yao90, HG91}.
They also imply lower bounds for branching programs,
time-space tradeoffs for Turing machines \cite{KN97}, and proof complexity
lower bounds \cite{BPS06}. The implications of good NOF lower bounds go in other, less expected directions as well. E.g., knowing the communication complexity of specific natural functions, even for $k=3$, would have profound implications in graph theory and combinatorics. Finally, the search for nontrivial protocols in this area is a wonderful challenge for algorithms designers. There is a short list of such beautiful examples \cite{CFL83, Gro94} which beg to be extended.

Furthermore, our understanding of NOF communication complexity, even for $k =3$ players, lags well behind our understanding of the standard  model ($k=2$ players). This gap is usually attributed to the dearth of proof techniques in the NOF setting. In the 2-party setting, many measures of complexity allow us to prove both upper and lower bounds. Such measures include matrix rank, various matrix norms, nonnegative rank, discrepancy, corruption bounds and information complexity. 
Most of these measures are computationally simple and admit dual characterizations which are very helpful in proving both upper and lower bounds.
On the other hand, in the NOF setting for $k \geq 3$, the key combinatorial objects are cylinder intersections (rather than combinatorial rectangles)
and tensor norms. These are far more complex, and thus far have resisted
a workable characterization. 

A case in point is the separation of randomized from deterministic communication complexity. The 2-party {\em equality function} has a randomized protocol of bounded cost, whereas a simple rank argument shows that every deterministic protocol must incur linear cost. This provides an optimal separation of deterministic and randomized communication complexity \cite{KN97}.
On the other hand, for $k \ge 3$, the best {\em explicit} separation between nondeterministic
and randomized NOF complexity is sub-logarithmic, even though
counting arguments yield linear separations \cite{BDPW07}. The \Exactn function is defined as follows: Input $x_1,\ldots,x_k \in [n]$ is accepted iff $\sum_i x_i=n$.
In their seminal paper, Chandra, Furst and Lipton \cite{CFL83} conjectured that \Exactn achieves a strong separation. They also found a relation between the communication 
complexity of this function and well-known problems in additive combinatorics and Ramsey theory. But thus far, despite considerable research effort, the lower bounds for \Exactn are much weaker even than the best (sub-logarithmic) explicit separations.

The main goal of our work is to further investigate the connections between NOF complexity of functions and questions in additive combinatorics, with the hope of stimulating further research to make progress in both directions. A large and rapidly growing body of work establishes interesting relationships
between problems in additive combinatorics and complexity theory. For example, the study of expander graphs and extractors, pseudorandomness, and property testing is closely related, some time even synonymous with similar notions in additive combinatorics. Moreover, techniques from complexity theory have been useful in additive combinatorics and vice versa. 
Some recent examples include the proof of the cap-set 
conjecture \cite{croot2016progression,ellenberg2017large} and Dvir's resolution \cite{dvir-kakeya} of the finite field Kakeya problem, both using the polynomial method, as well as the beautiful interplay between dense model theorems in additive combinatorics and the notions of boosting and hardcore sets from complexity theory. (For example, see the surveys \cite{Bibak13, Trevisan09, Lovett17} and the references therein.)

Here we consider a broad class of functions called high dimensional permutations. We uncover strong connections between the NOF communication complexity of these functions and several fundamental problems in additive combinatorics. 
Originally defined in \cite{li:lu}, a $(k-1)$-dimensional permutation is a function $f:[n]^k \to \{0,1\}$ such that
for every index $k \ge i \ge 1$ and for every choice of $x_1,\ldots,x_{i-1},x_{i+1},\ldots,x_k \in [n]$,
there is exactly one value of $x_i \in [n]$ for which $f(x_1,\ldots,x_{i-1},x_i,x_{i+1},\ldots,x_k)=1$.
This class of functions generalizes many well-studied functions in communication complexity. 
It is also closely related to many other functions such as the \Exactn function mentioned above.

We will show that many well-studied problems
in NOF complexity are not just related, but are in fact identical or nearly identical
to central problems in additive combinatorics. We feel that this mutual relation deserves much more attention, and that progress in this area is likely to greatly advance both domains.
Specifically we believe that the study of the communication complexity of high
dimensional permutations and related graph functions (defined in \cite{BDPW07}) is a worthwhile undertaking that will help us develop new lower bounds proof techniques for the notoriously
difficult NOF model.
Using these connections, we make modest progress on several upper and lower bounds in NOF communication
complexity.

\subsection{Our Contributions}

As mentioned above, our main goal and contribution is to unveil the strong relationships
between the NOF complexity of high dimensional permutation problems and central problems in
additive combinatorics and Ramsey theory.
Already the founding paper of Chandra, Furst and Lipton \cite{CFL83} makes a connection between the NOF complexity of
\Exactn and the areas of Ramsey theory and additive combinatorics.
A more general framework was introduced in \cite{beigel2006multiparty}: Given an abelian group $G$ and $T\in G$, 
the function \ExactT evaluates to $1$ on input $x_1,\ldots,x_k \in G$ iff $\sum_i x_i = T$ (this expression is well-defined since $G$ is abelian). The functions \ExactT are high dimensional permutations. (Note that this  holds as well for non-abelian $G$, though we need to specify the order at which $\prod_i x_i$ is evaluated).
Another strong connection 
is that the Hales-Jewett theorem, a cornerstone of Ramsey theory, can be interpreted in terms of communication complexity \cite{shraibman2017note}.

We establish a new and close connection between the
NOF communication complexity of high dimensional permutations and 
dense \rs\ graphs.
These graphs appear in various contexts in Combinatorics, Computer Science and Information Theory, thus
highlighting new connections between communication complexity and these
various problems. For example, an efficient deterministic
communication protocol for any permutation yields an efficient wiring scheme for shared
directional multi-channels. For more on this, see e.g., \cite{birk1993uniform} and \cite{alon2012nearly}.
In the classical, $k=2$ case, monochromatic submatrices play a key role in the theory. For higher $k$
this is replaced by the much more poorly understood monochromatic {\em cylinder intersection}. Naturally,
much of our work here revolves around these complicated objects. However, in certain simple cases we
are able to get a grip on the largest size of a cylinder intersection that contains only $1$-inputs of $f$.
As we show, in this case knowledge of this quantity essentially determines the NOF
communication complexity of $f$ (see more on this in the next section). The case in question is $k=3$ 
and the group $G = \mathbb{Z}_2^n$. As we show, the size of
the largest cylinder intersection containing only $1$-inputs of $f$ is the largest cardinality of a subset
$W\subseteq\mathbb{Z}_4^n$ such that for every three distinct members
$\mathbf{x},\mathbf{y},\mathbf{z}\in W$ there is an index $1\le i \le n$ for which $(x_i,y_i,z_i) \not\in X$, where
\[X = \{ (0,0,0), (1,1,1), (2,2,2), (3,3,3), (0,1,2), (1,0,3), (2,3,0), (3,2,1)\}.\]

This parameter may seem artificial, but in fact, this framework includes several important problems in combinatorics, for different choices of $X$. Thus, if we take 
$$X:=\{ (0,0,0), (1,1,1), (2,2,2), (3,3,3), (0,1,2) \},$$ 
then this becomes precisely the density Hales-Jewett problem, solved in \cite{furstenberg1991density}. Also, if $X$ is comprised of all triplets
$(a,b,c)\in \mathbb{Z}^3_4$ with $a + c = 2b$, we arrive at the cap-set problem for $\mathbb{Z}_4^n$ which was recently settled in breakthrough papers
by Croot, Lev, and Pach, and by Ellenberg and Gijswijt \cite{croot2016progression, ellenberg2017large}.
In the next subsection we list our new results that stem from these connections.

\medskip

\subsubsection{Upper Bounds}

We give a new algorithm for \Exactn as well as several other instances of \ExactT.
All previous upper bounds for these functions crucially depend on Behrend's famous construction \cite{behrend1946sets}
of a large set of integers with no 3-term arithmetic progressions.
This yields a large monochromatic cylinder intersection, and a simple probabilistic translation lemma then shows how to cover the whole space by monochromatic cylinder intersections, thus providing an efficient protocol. To show that this indeed yields a large monochromatic cylinder intersection, we appeal to the notion of {\em corner-free} sets \cite{ajtai1974sets, Solymosi2003}  which here, too, plays a key role.
We cannot realistically hope to improve the bounds by finding a construction better than Behrend's, in view of the many such failed attempts throughout the past 70 years (but note \cite{elkin2010improved}). However, Behrend's construction is actually more than we need. The solution of \ExactT only requires corner-free sets. That is, 3-term AP freeness implies corner-freeness, but we do not expect that the two concepts are equivalent.
We take a first step in this direction and give a new algorithm which is not dependent on 3-term AP freeness. We 
hope that this indicates a viable approach to improved protocols for the \Exactn function. 
We also obtain a nontrivial protocol for the $f^G_{3,T}$ problem for $G = \mathbb{Z}_2^n$.

Thus far the connections mentioned above are mostly applications of tools and results from Ramsey Theory and Additive Combinatorics such as Behrend's construction to questions in communication complexity, and hardly anything in the opposite direction. There are a few recent applications of communication complexity protocols to Ramsey Theory, such as
\cite{ada2015nof} and \cite{shraibman2017note}. In order to advance this line of research we need to develop algorithmic 
protocols for the underlying functions, and the protocols discussed above constitute a first step in this direction.

\subsubsection{Lower Bounds} 

We give a counting argument which shows that almost every $k$-dimensional permutation has communication complexity  $\Omega(\frac{\log n}{k})$. Clearly, up to the $\frac 1k$ factor, this is as high as this quantity can get.
Our proof relies on a recent lower bound of Keevash \cite{keevash2018existence} on the number of high-dimensional permutations. This method resembles the counting argument for graph functions of \cite{BDPW07}, which does not apply, though, to permutations.

Regarding bounds on explicit functions, we prove a weak upper bound on the size of a 1-monochromatic
cylinder intersection for any permutation (in fact our result holds for a wider family of functions that we call 
linjections). 
This bound uses a graph theoretic characterization of the communication complexity of permutations,
connecting it also to \rs\ graphs.
Not unexpectedly, our proof mirrors a similar result for \rs\ graphs: Solymosi \cite{Solymosi2003} showed that the multidimensional
Szemer\'{e}edi theorem follows from the triangle removal lemma.
We adapt Solymosi's proof to our context.
The main tools in the proof are thus the graph and the hypergraph removal lemmas.
The quantitative aspects of these lemmas are still poorly understood, so there is hope for possible future improvements here.

We note that previous results were limited to the \ExactT function for abelian groups with many factors, whereas
ours works for general permutations.

To emphasize the significance of the last point, consider the NOF complexity of following three classes of functions: (i) Permutations that come from Abelian groups, (ii) Those that come from general groups, (iii) Latin squares. We consider each such function up to an arbitrary renaming of rows and columns. The sizes of these three classes differ very substantially. For a given order $n$ the size of the relevant class is (i) $\exp(O(\sqrt{\log n}))$, (ii) At most \mbox{$\exp((\frac{2}{27}+o(1))\log^3 n)$,} and (iii) $((1+o(1))\frac{n}{e^2})^{n^2}$.

\medskip
For $k=3$ we can say more: The communication complexity of every $2$-dimensional  permutation $[n]^3 \to \{0,1\}$ is 
$$\Omega(\log \log \log n).$$

This extends the lower bound of \cite{beigel2006multiparty} from the realm of abelian groups
to all permutations.
The proof of the this lower bound
uses only elementary counting arguments, and is closely related to the result of \cite{graham2006monochromatic} on
monochromatic corners on the integer grid.

The above lower bound also implies a result of Meshulam that was derived toward the study of shared directional multi-channels.
Meshulam's result appears as Proposition~4.3 in \cite{alon2012nearly}, where further background can be found.

\subsection{Related Work}

The NOF model was introduced by Chandra, Furst and Lipton in \cite{CFL83}.
One of the functions they consider is \Exactn $:[n]^3\to\{0,1\}$. For $x,y,z \in [n]$, we let
\Exactn $(x,y,z)=1$ if and only if $x+y+z=n$.
Surprisingly, they proved that the communication complexity of this function is only $O(\sqrt{\log n})$, but their proof yields no explicit protocol. Although this function is not a permutation, it was observed in \cite{beigel2006multiparty},
the proofs go through as well if we work modulo $n$, in which case we deal with a $2$-dimensional permutation.
Thus far, this is the most efficient protocol found for any permutation. 

The protocol of \cite{CFL83} is based on Behrend's famous construction \cite{behrend1946sets}
of a large subset of $[n]$ with no three-term arithmetic progression. In addition, they prove an inexplicit
lower bound of $\omega_n(1)$ on the complexity of \Exactn. This is based on Gallai's result \cite[p.\ 38]{graham1990ramsey} that every finite coloring of a Euclidean space contains a monochromatic homoteth of every finite set in that space.

Beigel, Gasarch and Glenn \cite{beigel2006multiparty} have refined the study of \Exactn, and
considered the more general \ExactT problem. Here $G$ is an abelian group, $T$ is an element of $G$ and $k \ge 2$ an integer. In this scenario $k$ players need to decide whether $x_1+x_2+\cdots+x_k = T$, where the inputs $x_1,\ldots,x_k \in G$ are given to them in the NOF format. That paper showed that the communication complexity of $f^G_{3,T}$ is at least $\Omega(\log \log \log n)$ for every abelian group $G$ and any $T \in G$. 
For the case $G=\mathbb{Z}_n$, this follows as well from \cite{graham2006monochromatic} 
and a recent result of Shkerdov \cite{shkredov2009two} also yields a similar lower bound 
for every abelian group $G$.

For general $k\ge 3$ and for an abelian group $G$ that is the product of $t$ cyclic groups, it is shown in \cite{beigel2006multiparty} that the deterministic NOF complexity of \ExactT is $\omega_t(1)$.
The proof is by reduction to a lower bound from \cite{tessonapplication},
that is based on the Hales-Jewett Theorem (see \cite{graham1990ramsey}).
This lower bound is again not explicit, and yields only that the complexity is unbounded.

Note that \ExactT can be defined as well in non-abelian groups $G$. Namely, 
$f^G_{k,T}(x_1,\ldots, x_k)=1$ iff $x_1\cdot x_2\cdot\ldots\cdot x_k=T$, where now the order of multiplication matters.
Note also that the function \ExactT is a permutation for every group $G$, every
$T \in G$ and $k\ge 2$.

As mentioned above, \cite{BDPW07} studies graph functions and
give a nonexplicit strong separation between randomized and
deterministic NOF complexity. 
To be precise, this counting argument shows that most graph
functions $f: [n]^{k-1} \times [N] \to \{0,1\}$ with $N\cong \sqrt{\frac{n}{k}}$ have deterministic
communication complexity $\Omega(\log \frac{n}{k})$.
Still, even for $k=3$ it remains open to find {\em explicit} graph functions with
high deterministic communication complexity. Currently, the best lower bound on the deterministic communication complexity of a graph function $f: [n]^{k-1}\times  [N] \to \{0,1\}$
for $k\ge 3$ is $\Omega(\log \log n)$ proved in \cite{BDPW07} (using also results from \cite{BHK01}).
We note that for functions that are hard for randomized communication,
the discrepancy method has been used to establish NOF lower bounds (e.g., \cite{BNS92} and other papers.)
Unfortunately, the discrepancy
method cannot be used to the end of separating deterministic from randomized communication complexity,
since it applies to deterministic as well as to randomized communication complexity.

Lastly, we comment on the Hales-Jewett theorem, a pillar of Ramsey theory. It was previously
applied in the study of the combinatorial problems mentioned above. It turns out that this theorem has an equivalent formulation in the language of communication complexity \cite{shraibman2017note}, and is tightly coupled with the NOF multiparty communication complexity of high dimensional permutations.

\medskip

\noindent {\bf Organization.} In Section 2, we first give a brief survey of some of the key concepts and theorems  in additive combinatorics that we will be discussing throughout, and then define the NOF communication complexity model, graph, permutation functions and linjections, and some basic structural results about their complexity. In Section 3, we establish a connection between the NOF communication complexity of high dimensional permutations and \rs\ graphs, and give a special characterization for the case of $Z_2^n$ . Section 4 contains our new upper bounds, and Section 5 contains our new lower bounds. We conclude in Section 6 with many open problems and directions.

\section{Basics}

{\subsection{A Little Bit of Additive Combinatorics}

The basic tenet of this area is that the additive group $(\mathbb{Z},+)$ is not only an algebraic object, but carries as well a lot of combinatorial structure. 
The field is over a hundred years old and still teems with breakthroughs, fascinating open questions and lots of drama. 
Such a brief description cannot do justice to this area, and we refer the reader to the excellent book by Tao and Vu \cite{tao2006additive} 
and to several good online reviews. We start with van der Waerden's theorem
\cite{van1927beweis} from 1927: 
\begin{quotation}
{\it For every $r$ and $k$ and for every large 
enough $N$, if the elements of $[N]:=\{1,\ldots,N\}$ 
are colored by $r$ colors, then there must exist a length-$k$ monochromatic arithmetic progression}. 
\end{quotation}

This remarkable theorem suggests many further avenues of research. E.g., shouldn't {\it the most abundant color} 
necessarily contain long arithmetic progressions? This has led Erd\H{o}s 
and Tur\'an \cite{erdos1936some} to pose several 
questions some of which have already found spectacular solutions and some still open. Answering one of their questions, 
Szemer\'edi famously proved \cite{szemeredi1975sets} in 1975\footnote{For this discovery he was awarded the 2012 Abel Prize.} 
\begin{quotation}
{\it For every $\epsilon>0$ and every integer $k$ there is an $n_0$ such that if $N > n_0$, every subset of $[N]$ 
of cardinality $\ge \epsilon N$ must contain a $k$-term arithmetic progression}. 
\end{quotation}
Key to Szemer\'edi's proof is his 
Regularity Lemma (SzRL). Considerable research effort is dedicated to the challenging question of determining how $N$ depends 
on $k$ and $\epsilon$. Another crowning achievement in this area is Green and Tao's theorem \cite{green2008primes} of 2004:
\footnote{This was one of the achievements for which Tao won the 2006 Fields Medal.} 
\begin{quotation}
{\it There are arbitrarily 
long arithmetic progressions comprised only of prime numbers}. 
\end{quotation}

The special case of Szemer\'edi's Theorem dealing with {\em arithmetic triples} had already been proved by 
Roth \cite{roth1953certain} in 1953 using harmonic analysis.\footnote{One of several reasons for his 1958 Fields Medal.} 
Here, too, the quantitative side of the theorem is far from being resolved. The best {\em lower bound} that 
we have comes from a 1946 construction of Behrend \cite{behrend1946sets} - 
A subset of $[N]$ of density $\exp(-c\cdot\sqrt{\log N})$ that contains no arithmetic triple.

These theorems are in some well-defined sense one-dimensional. What about higher dimensions? 
A subset of vectors $\vec{a}_1,\ldots,\vec{a}_n \in [n]^d$ forms a 
{\em combinatorial line} if for every 
$d\ge i\ge 1$ either $\vec{a}_1,\ldots,\vec{a}_n$ all have the same $i$-th coordinate or their $i$-th coordinates
contain every value in $\{1,\ldots,n\}$. In 1963 Hales and Jewett (HJ) \cite{hales2009regularity} proved:
\begin{quotation}
{\it  For every $k$ and $n$ there is $d$ such that in every coloring of $[n]^d$ by $k$ colors there must be 
a monochromatic combinatorial line}.
\end{quotation}
The standard text of this field \cite{graham1990ramsey} places this theorem as the 
cornerstone of Ramsey Theory. The quantitative aspect of the HJ theorem is still poorly understood. The 
{\em density version} of the HJ theorem is related to the HJ theorem in the same way that Szemer\'edi's theorem is related to the van der Waerden's theorem. This was established in Furstenberg and Katznelson 
\cite{furstenberg1991density} in 1991 
using tools from {\em ergodic theory}. They also proved the {\em high-dimensional Szemer\'edi theorem}:
\begin{quotation}
{\it  Given a finite $S\subset \mathbb{Z}^d$ and $\delta>0$, if $n$ is large enough, then every subset of $[n]^d$ of cardinality $\delta n^d$ must contain a {\em homothet} of $S$, i.e., a subset of the form 
$a\cdot S + \vec b$ for some integer $a$ and $\vec b\in \mathbb{Z}^d$}. 
\end{quotation}
By nature, ergodic-theoretic 
methods provide {\em no} quantitative bounds, and only in 2012 did a {\em polymath} group \cite{polymath2009new} 
find combinatorial proofs for these theorems.

One of the earliest applications of SzRL is {\em the Ruzsa-Szemer\'edi (6,3)-theorem} \cite{ruzsa1978triple}:
\begin{quotation}
{\it A 3-uniform hypergraph on $n$ vertices in which no $6$ vertices contain $3$ edges has at most $o(n^2)$ edges}. 
\end{quotation}
The (6,3)-theorem is essentially equivalent to a weak version of the {\em triangle removal lemma} :
\begin{quotation}
{\it  For every $\epsilon > 0$ 
there is a $\delta > 0$ such that for large enough $n$, every order-$n$ graph with fewer than $\delta n^3$ 
triangles can be made triangle-free by removing $\epsilon n^2$ edges}. 
\end{quotation}
The quantitative aspects of these 
theorems remain poorly understood. It was observed by Solymosi \cite{Solymosi2003} that the Triangle Removal Lemma 
implies the {\em Corners Theorem} of Ajtai and Szemer\'edi \cite{ajtai1974sets}: 
\begin{quotation}
{\it For every $\delta>0$ and 
large enough $n$, every subset of $[n]\times[n]$ of cardinality $\delta n^2$ must contain three elements of 
the form $(a,b), (a+d, b), (a, b+d)$}.
\end{quotation}
This theorem easily yields Roth's theorem.

\subsection{NOF Communication Complexity}
\label{sec_cc-basics}

In the Number On the Forehead (NOF) multiparty communication complexity game, $k$ players collaborate to compute a function $f: X_1 \times \ldots \times X_k \rightarrow \{0,1\}$. Usually, $X_i = [n]$ for all $i \in [k]$, but we also consider occasionally a variation where the last player is exceptional and $X_k=[N]$ for some integer $N$ that is not necessarily equal to $n$.

For $(x_1,\ldots,x_k) \in X_1 \times \ldots \times X_k$, and
for each $i \in [k]$, player $i$ receives $x^{-i} \in X_1 \times \ldots \times X_{i-1}\times  X_{i+1} \times \ldots \times X_k$; that is, all but $x_i$. 
The players exchange bits according to an agreed-upon protocol, by writing
them on a publicly visible blackboard. The protocol specifies, for every possible
blackboard contents, whether or not the communication is ongoing. It shows the final output when the communication is over, and shows the next player to speak if the communication is still ongoing.
The protocol also specifies what each player writes as a function
of the blackboard contents and of the inputs seen by that player.
The {\it cost} of the protocol is the maximum number of bits written on the blackboard.

The {\it deterministic communication complexity} of $f$, $D_k(f)$, is the
minimum cost of a deterministic protocol for $f$ that always outputs the
correct answer.
A randomized protocol of cost $c$ is just a distribution over deterministic protocols
each of cost at most $c$.
For $0 \leq \epsilon < 1/2$,
the {\it randomized communication complexity} of $f$, $R_{k,\epsilon}(f)$,
is the minimum cost over randomized protocols such that
for every input, err with probability at most $\epsilon$
(over the distribution of deterministic protocols).

In the $k=2$ players case, the key combinatorial objects of study
are combinatorial rectangles: Every cost-$c$ communication protocol for 
$f: X_1 \times X_2 \rightarrow \{0,1\}$ partitions
$X_1 \times X_2$ into $2^c$ monochromatic combinatorial rectangles.
For $k$-party NOF communication, {\it cylinder intersections} take center stage:

\begin{definition}
A {\it cylinder in dimension} $i$ is a subset $S\subseteq\prod X_i$ such that if $(x_1,\ldots,x_k) \in S$, then $(x_1,\ldots,x_{i-1},x_i',x_{i+1},\ldots,x_k) \in S$  for all $x_i'$.
A {\it cylinder} {\it intersection} is a set of the forn $\cap_{i=1}^k T_i$, where $T_i$ is a cylinder in dimension $i$.
\end{definition}

The higher-dimensional counterpart of the above statement says that a cost-$c$ NOF communication protocol for $f :X_1 \times \ldots X_k \rightarrow \{0,1\}$ induces a partition of $X_1 \times \ldots \times X_k$ into $2^c$ {\it monochromatic cylinder intersections}. Here is how the argument starts: Suppose that player $i$ is the first to communicate. The input set of all other players is split in two: That set of values on which $P_i$ sends zero resp.\ one in the first transmission. This partitions the $k$-dimensional binary cube into two cylinders. The rest of the argument is routine.

\subsection{Graph Functions, Permutations and Linjections}

\begin{definition}
The line $L \subseteq [n]^k$, defined by a pair $(a,i)$, where $a \in [n]^{k-1}$, $i \in [k]$, is the set of vectors $v \in [n]^{k}$ such that $v^{-i}=a$ and $v_i$ is an arbitrary element in $[n]$.
\end{definition}

\begin{definition} 
A function $f: [n]^{k-1}\times [N] \to \{0,1\}$ is a {\em graph
function} if for every $(x_1,\ldots,x_{k-1})$ there is a unique $b \in [N]$ such that
$f(x_1,\ldots,x_{k-1},b)=1$. In other words, every line in the $k^{th}$ dimension, $L=(a,k)$, intersects $f^{-1}(1)$ in exactly one point.
\end{definition}

Associated with every graph function $f: [n]^{k-1} \times [N] \rightarrow \{0,1\}$ is a map $A(f):[n]^{k-1} \to [N]$,  where $A(f)(x_1,\ldots,x_{k-1})=y$ if and only if $f(x_1,\ldots,x_{k-1},y)=1$. We consider the two as one and the same object and freely switch back and forth between the two descriptions.

\begin{definition}
Let $f: [n]^{k-1} \times [N] \to \{0,1\}$ be a graph function. We denote by $\ind_k(f)$ the largest size of a cylinder intersection that is contained in $f^{-1}(1)$. In other words, the largest cardinality of $1$-monochromatic cylinder intersection with respect to $f$. Also, let  $\chr_k(f)$ be the least number of $1$-monochromatic cylinder intersections whose union is $f^{-1}(1)$. We omit the subscript $k$ when it is clear from context.
\end{definition}

Given a graph function $f$, the measure $\chr(f)$ corresponds to the nondeterministic NOF communication complexity
of $f$, since it is a covering of the 1's of $f$ by cylinder intersections \cite{KN97}.  
In general, the nondeterministic NOF communication complexity
of a Boolean function can be much smaller than the deterministic complexity -- in fact, for the set disjointness function,
nondeterministic complexity is logarithmic in the deterministic complexity (for constant $k$). However, graph functions are special;
the following lemma shows that for graph functions, the two notions basically coincide.
The proof is an adaptation of a proof from \cite{CFL83};
see also \cite{beigel2006multiparty, BDPW07} for similar arguments.

\begin{theorem}
\label{D_chr_gf}
For every graph function $f:[n]^{k-1} \times [N] \to \{0,1\}$,
$$
\log \chr_k(f) \le D_k(f) \le \lceil \log \chr_k(f) \rceil + k - 1.
$$
\end{theorem}

\begin{proof}
The lower bound follows from
the fact that every $c$-bit communication protocol for a function $f$ partitions the input space
into at most $2^c$ cylinder intersections that are monochromatic with respect to $f$ (see \cite{KN97} for more details).

To prove the upper bound $D_k(f) \le \lceil \log \chr_k(f) \rceil + k - 1$, fix
a covering of the 1's of $f$ by $\chr_k(f)$ many $1$-monochromatic cylinder intersections, and associate
with each cylinder intersection a number/name in $[\chr_k(f)]$.
Consider the following protocol for $f$:
On input $x_1,x_2,\ldots,x_{k-1},y$, the last player, $P_k$, computes
$y'$ such that $f(x_1,x_2,\ldots,x_{k-1},y')=1$ and writes the name, $b$,
of the cylinder intersection containing $(x_1,x_2,\ldots,x_{k-1},y')$ on the board.
Since $f$ is a graph function $y'$ is unique, and so is $b$.
Then for each $i=1,\ldots,k-1$, player $P_i$ checks whether there is a value $x_i'$ such that 
$(x_1,x_2,\ldots,x_{i-1},x_i',x_{i+1},..,y)$ is in the $b^{th}$ cylinder intersection;
if so, $P_i$ writes a $1$ on the board, and otherwise writes $0$.
(In other words, each player checks whether this cylinder intersection $b$ could be consistent
with their view of the input.) The protocol outputs $1$ if and only if
all players write $1$'s on the board.

The total number of bits communicated in this protocol is $\lceil \log \chr_k(f) \rceil + k - 1$. We turn to prove that the
protocol is correct. When $f(x_1,x_2,\ldots,x_{k-1},y)=1$, the protocol clearly outputs $1$.
In the other direction, suppose that $f(x_1,\ldots,x_{k-1},y)=0$, but the protocol outputs 1.
Let $y'$ be the unique value such that $f(x_1,\ldots,x_{k-1},y')=1$ and suppose that
$(x_1,\ldots,x_{k-1},y')$ is in cylinder intersection $b$ so $P_k$ writes $b$ on the
blackboard, and then Players $P_1,\ldots,P_{k-1}$ all write $1$'s.
Then there exist $x_1',x_2',\ldots,y'$ for which
\[ \{(x_1',x_2,\ldots,x_{k-1},y),(x_1,x_2',\ldots,x_{k-1},y), \ldots , (x_1,x_2,\ldots,x_{k-1},y')\} \]
are all in the $b^{th}$ cylinder intersection, and therefore they are all in $f^{-1}(1)$.
By definition of a cylinder intersection, this implies that $(x_1,\ldots,x_{k-1},y)$ is
also in $b$, and therefore it is also in $f^{-1}(1)$.
\end{proof}

A communication protocol in which players write only one message on the board, of arbitrary length is called a {\it one-way} protocol. Note that this applies the protocol in the above proof. The permission to send messages of arbitrary length may make one-way protocols much more powerful than standard protocols \cite{nisan1991rounds, BHK01}. However for graph functions, one-way protocols and regular protocols are equally powerful:

\begin{corollary}
For every graph function $f:[n]^{k-1} \times [N] \to \{0,1\}$ there holds\\
$D_k(f) \le D^1_k(f) \le D_k(f) + k$
where $D^1_k(f)$ is the one-way communication complexity of $f$.
\end{corollary}

Implicit in the above proofs is the fact that for
graph functions, monochromatic cylinder intersections can be nicely characterized by forbidden (dual) objects called {\it stars}, which we define next.
We will see in the next section that stars are very closely connected
to corners (and higher dimensional generalizations) in Ramsey theory.

\begin{definition}
A {\em star} $Star(\mathbf{x},\mathbf{x}')$ is a subset of
$[n]^{k-1}\times [N]$ of the form
\[\{(x'_1, x_2, \ldots, x_k), (x_1,x'_2, \ldots, x_k), \ldots, (x_1,x_2,\ldots, x'_k)\},\]
where $x_i \ne x'_i$ for each $i$. We refer to $\mathbf{x}=(x_1, x_2, \ldots, x_k)$ as the star's {\em center},
and note that the center does {\em not} belong to the star. 
\end{definition}

\begin{lemma}
\label{lem_cyl_int_gf}
Let $f: [n]^{k-1}\times [N] \to \{0,1\}$ be a graph function, and let $S \subseteq f^{-1}(1)$. Then $S$ is a ($1$-monochromatic) cylinder intersection with respect to $f$ if and only if it does not contain a star.
\end{lemma}

\begin{proof}
It is not too hard to see \cite{KN97} that a subset $ S \subseteq [n]^{k-1} \times [N]$ is a cylinder intersection
if and only if for every star that is contained in $S$,
the center is also in $S$.
Thus for any $S$, if $S$ does not contain a star, then $S$ is a cylinder intersection.
For the other direction, let $S \subseteq f^{-1}(1)$, and
suppose that $S$ contains a star.
Then by definition of a graph function, $f(x_1,x_2,\ldots,x_k)=0$, where
$(x_1, x_2, \ldots, x_k)$ is the center of the star. Therefore
$S$ does not contain the center of star (so $S$ is not a cylinder
intersection.)
\end{proof}

Next we define high dimensional permutations and linjections.
 
\begin{definition}
A $(k-1)$-dimensional permutation of order-$n$ is a map $f:[n]^k \rightarrow \{0,1\}$ with the property
that for every line $L=(a,i)$ in $[n]^k$, $\left| L \cap f^{-1}(1) \right| =1$. 
\end{definition}

In other words, $f$ is a permutation function if and only if every line contains a unique 1 entry.
This property is easily seen to be equivalent to the property that for every choice of
$x_1,\ldots,x_{i-1},x_{i+1},\ldots,x_k \in [n]$, there is exactly one value,
$A_i(x^{-i})$ for $x_i \in [n]$ such that 
$$f(x_1,\ldots,x_{i-1},A_i(x^{-i}),x_{i+1},\ldots,x_k)=1.$$

\begin{example} 
\label{a-star-latin}
For the sake of gaining better intuition we often consider the important special case $k=3$. This is insightful, since $2$-dimensional permutations $f:[n]^3 \to \{0,1\}$ are synonymous with Latin squares. In this case $A(f)$ is an $n\times n$ matrix with entries in $[n]$ where every row and column contains each element in $[n]$ exactly once. Here we see an elementary but important connection with additive combinatorics; stars coincide with the well-studied notion of {\it corners}
\cite{ajtai1974sets, Solymosi2003}. 
A  star is a triplet of entries in $f^{-1}(1)$,  $(x,y,z'), (x',y,z), (x,y',z)$,
which corresponds to the "corner" or "$A$-star",  $(x,y),(x',y),(x,y')$, 
where $A(x',y) $ and $A(x,y')$ have the same value ($z$), but $A(x,y)$ has a different value ($z'$).
\end{example}

\begin{example}  
High dimensional permutations generalize the family of
functions \ExactT for abelian groups $G$.
For this communication problem, each player receives (on his/her forehead) an element $x_i \in G$,
and they want to decide whether or not $x_1 + \ldots + x_k =T$, that is, whether the sum of the
elements is exactly $T$.
\end{example}

We can further generalize the notion of a permutation function as follows. 

\begin{definition}
A {\em linjection} is a graph function $f: [n]^{k-1}\times [N] \to \{0,1\}$ with $N \ge n$
such that $|f^{-1}(1)| = n^{k-1}$ and every line contains {\bf at most} one point at which $f=1$.
A function $f$ is a linjection if and only if the restriction of $A(f)$ to any line is an injection.
\end{definition}

Linjections are graph functions where $N \geq n$, (with permutation functions
corresponding to $n=N$), but not vice versa.

There are very simple graph functions of bounded communication complexity. As we will see, this is not the case for permutations and linjections. Determining the least possible communication complexity of a linjection in certain dimensions
is an interesting and challenging problem.  Henceforth, we use and study the following two notions.

\begin{definition}
Define $\ind_k(n,N) = \max_f \ind_k(f)$, and $\chr_k(n,N) = \min_f \chr_k(f)$, both taken over all linjections
$f: [n]^{k-1} \times [N] \to \{0,1\}$. 
\end{definition}

Note that $\chr_k(n,N) \ge n^{k-1}/\ind_k(n,N)$.

\section{High Dimensional Permutations and Additive Combinatorics}

\subsection{A Graph-theoretic Characterization}
\label{sec:graph_theoretic_characterization}

In this section we give a new characterization of
$\alpha_k$ which will turn out to be a variant of the maximum density of \rs\ graphs.
We start with the case $k=3$.

Recall that we can view a linjection $f: [n]^2 \times [N] \to \{0,1\}$ as an $n \times n$ matrix, $A=A(f)$ with entries from $[N]$. Alternatively we view it as a tripartite graph $G(A)$ with parts $R=[n], C=[n]$ and $W\subseteq [N]$. Its edge set is defined as follows: for every triple $(x,y,b) \in f^{-1}(1)$,  we add the triangle $(x,y),(y,b),(x,b)$, $x \in R$, $y \in C$, $b \in W$ to $G(A)$. In particular, $R\cup C$
span  a complete bipartite subgraph of $G(A)$ and $(i,b)$,  $i\in R$, $b\in W$ is an edge iff there is a $b$ entry in row $i$ of $A$, likewise for columns.

\renewcommand{\rho}{\overline{\ind}}

Let us consider the triangles  $<x,y,b>$, $x \in R, y \in C, b \in W$, in $G(A)$.
A  triangle $<x,y,b>$  in $G$ is  {\em trivial} if $A(x,y)=b$. However, there can also
be nontrivial (induced) triangles in $G$, which correspond to centers of stars.
We define a $G$-{\em star} to be a
triple of triangles in $G$  of the form \[<x,y,b'>, <x',y,b>, <x,y',b>.\] The point is that
while these (trivial) triangles are edge-disjoint, their union contains the additional induced triangle $<x,y,b>$.
Define $\rho(G)$ to be the largest cardinality of
a family of edge-disjoint triangles that contains no $G$-star. In other words, a family of
edge-disjoint triangles the union of which contains no additional triangle.

Let $\rho(n,N) = \max_G \rho(G)$ where the maximum is over subgraphs of $K_{n,n,N}$. Then:
\begin{theorem}
\label{equiv-rho-kappa}
For every two integers $n,N >0$,  if $n\le N$  then $\ind_3(n,N) \le \rho(n,N)$. If $N \ge 2n-1$,
then $\ind_3(n,N) = \rho(n,N)$.
\end{theorem}

\begin{proof}
We show first that $\ind_3(n,N) \le \rho(n,N)$. Let $f: [n] \times [n] \times  [N] \to \{0,1\}$  be a linjection and let
$S \subseteq [n]\times[n] \times [N]$ be a star-free subset of $f^{-1}(1)$. We prove the claim by constructing a $G$-star-free family $T$ of $|S|$ edge-disjoint triangles in $G=G(A(f))$. Let
\[T = \{<x,y,b>|(x,y,b) \in S\}.\]
The claim follows, since stars $\{(x',y,b),(x,y',b),(x,y,b)\}$  correspond to $G$-stars in $T$.
Next we prove the reverse inequality $\ind_3(n,N) \ge \rho(n,N)$ when $N \ge 2n-1$.

Given a $G$-star-free family $T$ of edge-disjoint triangles
in a subgraph $G$ of $K_{n,n,N}$, we find a
linjection $A: [n]\times [n] \to [N]$ that contains an $A$-star-free
subset $S \subset [n]^2$ of size $|T|$.
In the proof we actually first construct $S$ and only then proceed to define $A$ in full.

We define $S$ to be the projection of $T$ to its first two coordinates. Namely,
\[
S = \{(x,y) ~|~  <x,y,b> \in T\text{~for some~}b\}.
\]
To define $A$, we first let $A(x,y)=b$ for every $<x,y,b>\; \in T$.

Since $T$ is $G$-star-free, it follows that $S$ is $A$-star-free. What is missing is that $A$
is only partially defined. We show that when $N \ge 2n-1$ this partial definition can be extended to a linjection.
Since the triangles in $T$ are edge-disjoint it follows that in the partially defined $A$, no value
appears more than once in any row or column.
It remains to define $A$ on all the entries outside of $S$ and maintain this property.
Indeed this can be done entry by entry. At worst there are $2n-2$ values that are forbidden for the
entry of $A$ that we attempt to define next, and therefore there is always an acceptable choice.
\end{proof}

\medskip

\noindent{\bf General $k$.} The construction for general $k$ is a natural extension of the case $k=3$.
We associate with every linjection $A: [n]^{k-1} \to [N]$ a $k$-partite $(k-1)$-uniform hypergraph $H(A)$.
The parts of the vertex set are denoted $Q_1,\ldots,Q_{k-1}$ and $W$. Each $Q_i$ is a copy of $[n]$ and, as above,
$W$ is the range of $A$. There is a complete $(k-1)$-partite hypergraph on the \mbox{$k-1$} parts $Q_1,\ldots,Q_{k-1}$. Given
\mbox{$x_1\in Q_1,\ldots,x_{i-1}\in Q_{i-1}, x_{i+1}\in Q_{i+1},\ldots,x_{k-1}\in Q_{k-1}$} and $w\in W$, we put
the hyperedge $x_1,\ldots,x_{i-1}, x_{i+1},\ldots,x_{k-1}, w$ in $H(A)$ iff there is a (necessarily unique) $x_i^{\ast}\in [n]$
for which $A(x_1,\ldots,x_{i-1},x_i^{\ast}, x_{i+1},\ldots,x_{k-1})= w$.

We proceed to investigate {\em cliques} in $H(A)$, i.e., sets of $k$ vertices, every $k-1$ of which form an edge.
For $k=3$, we distinguished between those triangles in $G(A)$ that correspond to an entry in $[n]^2$ and those that form a star, and a similar distinction applies for general $k$.

It is easy to see that if $A(x_1,\ldots,x_{k-1})= w$, then
$x_1,\ldots,x_k,w$ from a clique. Such a clique is considered {\em trivial}.
In contrast, $x_1,\ldots,x_{k-1},w$ is a nontrivial clique iff for every $i$ there exists an $x_i'\ne x_i$ such that
\mbox{$A(x_1,\ldots,x_{i-1},x_i', x_{i+1},\ldots,x_{k-1})= w$}.

As above, we define for $H=H(A)$ the parameter
$\rho_k(H)$. It is the largest size of a family $K$ of cliques in $H$
such that: (i) No two share a hyperedge, and (ii) The hypergraph comprised of all cliques in $K$ contains no additional cliques.
Let $\rho_k(n,N) = \max_H \rho_k(H)$ over all $k$-partite $(k-1)$-uniform hypergraphs $H$.
Then
\begin{theorem}
\label{equiv-rho-kappa-general-k}
For every two integers $n\le N$ there holds $\ind_k(n,N) \le \rho_k(n,N)$, and if $N > (k-1)(n-1)$
then $\ind_k(n,N) = \rho_k(n,N)$.
\end{theorem}

\begin{proof}
It is not hard to check that a family of hyperedge-disjoint cliques induces an additional clique if and only if
it contains $k$ cliques of the form:
$$
<x_1,\ldots,x_{k-1},b'>,<x_1',\ldots,x_{k-1},b>, \ldots,<x_1,\ldots,x_{k-1}',b>.
$$
We call such a set of cliques an {\em $H$-star}.

The proof is similar to the proof of Theorem~\ref{equiv-rho-kappa}. 

Let $f: [n]^{k-1} \times [N] \to \{0,1\}$ be a linjection, and let $S$ be a star-free subset of $[n]^{k-1}\times [N]$.
Define the following family of (trivial) cliques in $H=H(A)$:
$$
K = \{<x_1,\ldots,x_{k-1},b> ~|~ (x_1,\ldots,x_{k-1},b) \in S\}.
$$
Since the cliques in $K$ are trivial, they are hyperedge-disjoint. Also, since
$S$ is star-free, it follows that $K$ contains no $H$-stars, as $H$-stars directly correspond to stars
in $[n]^{k-1} \times [N]$.

For the reverse inequality, given a family $K$ of edge-disjoint cliques with no $H$-stars
in the complete $k$-partite $(k-1)$-uniform hypergraph, define similarly
\[
S = \{(x_1,\ldots,x_{k-1},b) ~|~  <x_1,\ldots,x_{k-1},b> \in K\}.
\]
and let $f(x_1,\ldots,x_{k-1},b)=1$ for every $(x_1,\ldots,x_{k-1},b)\; \in S$.

Since the cliques in $K$ are pairwise edge-disjoint, $S$ contains at most one entry
in every line, and $S$ is star-free since $K$ does not contain a $H$-star.

It remains to show that $f$ can be extended to a linjection when $N > (k-1)(n-1)$.
We omit this argument which is similar to the proof of
Theorem~\ref{equiv-rho-kappa} and only note that it is better to formulate it in terms of $A=A(f)$.
\end{proof}

The proofs of Theorem~\ref{equiv-rho-kappa-general-k} and~\ref{equiv-rho-kappa} make it interesting to better
understand the relationship between $\rho_k(n,N)$ and $\ind_k(n,N)$. As the proofs show, $\rho_k(n,N)$ is the
largest cardinality of a star-free subset of $[n]^{k-1}\times[N]$ that meets every line in $[n]^{k-1}\times[N]$ at most once.
To qualify for $\ind_k(n,N)$ this subset must, in addition, be extendable to a linjection,
so clearly $\rho_k(n,N)\ge\ind_k(n,N)$. We wonder whether this additional requirement creates a substantial difference between the two
parameters. Specifically, how are $\rho_k(n,N)$ and $\ind_k(n,N)$ related in the range $n\le N \le (k-1)(n-1)$?
These two parameters need not be equal in this range, since $\ind_3(4,4)=8$ and $\rho_3(4,4)=9$, as we show
in Section 4.3.

\medskip

\noindent{\bf Connection to \rs\ Graphs.} A graph is called an $(r,t)$-\rs\ graph if its edge set can be partitioned into $t$
edge-disjoint induced matchings, each of size $r$. These graphs were introduced in 1978 and
have been extensively studied since then. Of particular interest are
dense \rs\ graphs, with $r$ and $t$ large, in terms of $n$, the number of vertices. Such graphs have applications in
Combinatorics, Complexity theory and Information theory.
Also, there are several known interesting constructions, relying on different techniques.

Let $G$ be a tripartite graph with parts $R, C, W$ of cardinalities $n, n, N$
respectively. Let $T$ be a $G$-star-free family of edge disjoint triangles
in $G$. Let $F$ be the bipartite graph with parts $R$ and $C$ where there is an edge between $r\in R$ and $c\in C$
iff there is some $b\in W$ such that $(r,c,b)\in T$.
Then $F$ is the union of at most $N$ edge-disjoint induced matchings, since all the edges
that correspond to a given $b \in W$ form an induced matching.

This construction can easily be reversed: Let $F$ be a subgraph of $K_{n,n}$ that is the union of $N$ edge disjoint induced matchings,
with a total of $\rho$ edges. We can construct a tripartite $G$ (a subgraph of $K_{n,n,N}$) that contains
a family of $\rho$ pairwise disjoint triangles, and has no $G$-stars. We conclude that

\begin{observation}
Let $n \le N$ be positive integers, then $\rho_3(n,N)$ is the largest number of edges
in a union of $N$ edge-disjoint induced matchings in $K_{n,n}$.
\end{observation}

This observation exhibits a strong connection between (i) The problem of constructing dense $(r,t)$-\rs\ graphs,
and (ii) The construction of a large star-free subset $S \subseteq [n]\times[n]\times[t]$ that meets
every line at most once. The two problems differ only slightly.
In one, the underlying graph is bipartite and in the other all induced matching must have the same cardinality.
But these differences can be bridged quite easily, as observed in the following lemma.

\begin{lemma}
\begin{enumerate}
\item
If there exists an $(r,t)$-\rs\ graph on $n$ vertices, then $\rho_3(\frac{n}{2},t) \ge \frac{rt}{2}$.
\item
If $\rho_3(n,t) \ge rt$ then there exists a $(\frac{r}{2},t)$-\rs\ graph on $n$ vertices.
\end{enumerate}
\end{lemma}

\begin{proof}
For the first claim, let $G = (V,E)$ be a $(r,t)$-\rs\ graph on $n$ vertices, and
let $E_1,E_2,\ldots, E_t$ be the partition of $E$ into induced matchings. We can find (e.g., by a random choice) a subset $A\subset V$
of $\lfloor\frac n2\rfloor$ vertices, so that at least $|E|/2$ edges are
in the cut $C=(A, \bar{A})$.
Also, $C \cap E_1, C \cap E_2,\ldots, C \cap E_t$ is a partition of the edges of
the bipartite graph $(A, \bar{A}, C)$ into $t$ disjoint induced matchings.
Therefore, $\rho_3(\frac{n}{2},t) \ge \frac{rt}{2}$.

For the second part, suppose that $\rho_3(n,t) \ge rt$. Namely, there is a collection
of disjoint induced matchings $M_1,\ldots M_t \subseteq E(K_{n,n})$ with $\sum_1^t |M_i|\ge rt$. We split each $M_i$ into $\lfloor\frac{2|M_i|}{r}\rfloor$ sets of $\ge r/2$ edges each. Note that $\sum_1^t a_i \ge rt$ implies that $\sum_1^t \lfloor\frac{2a_i}{r}\rfloor \ge t$ and a subset of an induced matching is an induced matching, so we finally
have a family of at least $t$ disjoint induced matchings each of size $\frac{r}{2}$.
\end{proof}

\subsubsection{Application to Shared Directional Multi-channels}
\rs\ graphs have various applications in several fields \cite{Solymosi2003, alon2012nearly, ruzsa1978triple, alon2002testing, haastad2003simple, birk1993uniform}.
In \cite{birk1993uniform} they are applied to Information Theory, and the study of {\em shared directional multi-channels}, a subject that is strongly related to communication complexity.
Such a channel is comprised of a set of inputs and a set of outputs
to which are connected transmitters and receivers respectively. Associated with each input is a set of outputs, that
receive any signal placed at that input.
A message is received successfully at an output
of the channel if and only if it is addressed to the receiver connected to that output and no other
signals concurrently reach that output. Therefore, when communicating over a shared channel,
we want the edges (corresponding to messages sent in one round) to form an induced matching.
The challenge is to partition $K_{n,n}$ into families of pairwise disjoint induced matchings.
The number of parts correspond to the number of receivers allowed at each output, and
the number of matchings in each partition corresponds to the number of rounds.

The relation to communication complexity is as follows: A $c$-bit communication protocol for
any linjection $A: [n]\times[n] \to [N]$ induces a partition of $K_{n,n}$ into $c$ such families of disjoint induced matchings.
Thus, such a communication protocol, gives an $N$ round protocol for the shared directional multi-channel, with $c$ receivers
per station, and vice-versa.

In constructing a shared directional multi-channel, we seek to minimize the number of rounds required for
a given number of transmitters.
Alon, Moitra, and Sudakov \cite{alon2012nearly} showed that for any $\epsilon > 0$ there is partition of $K_{n,n}$ into at most $2^{O(\frac{1}{\epsilon})}$
graphs each of which is a family of at most $O(n^{1+\epsilon})$ induced matchings. This gives an
$O(n^{1+\epsilon})$ round protocol for shared directional multi-channel with $2^{O(\frac{1}{\epsilon})}$ receivers.

Translated to the language of NOF protocols and combining with Corollary~\ref{cor:chr_N} (see
Section~\ref{sec:a_lower_bounds_on_chr_3_n_N} in the sequel), we conclude:

\begin{theorem}
For all $\epsilon >0$ and all large enough $n$, there holds:
\[2^{O(\frac{1}{\epsilon})} \ge
\chr_3(n,n^{1+\epsilon}) \ge \Omega(\log \frac{1}{\epsilon}).\]
\end{theorem}

\subsection{A Characterization of $\ind_k(f^{\mathbb{Z}_2^n}_{k,T})$}

In this section we focus on the problem \ExactT for the abelian group $\mathbb{Z}_2^n$. In other words, we study
the permutation $f^{\mathbb{Z}_2^n}_{k,T}$. 
We give an alternative characterization of 
$\ind_3(f^{\mathbb{Z}_2^n}_{3,T})$ which brings forth the relation
between this problem and several known combinatorial objects.
The complexity of $f^{\mathbb{Z}_2^n}_{k,T}$ is independent of $T$, so we will omit the subscript $T$ in this section. 
Also, throughout this section we let $A^{G}_{k} = A(f^{G}_{k})$.

Let $X \subset \mathbb{Z}_4^3$. We call a subset of $W\subseteq \mathbb{Z}_4^n$ {\em $X$-free} if 
for every three distinct members
$\mathbf{x},\mathbf{y},\mathbf{z}\in W$ there is an index $1\le i \le n$ for which $(x_i,y_i,z_i) \not\in X$.
\begin{theorem}
\label{th:alt_char}
Let \[X = \{ (0,0,0), (1,1,1), (2,2,2), (3,3,3), (0,1,2), (1,0,3), (2,3,0), (3,2,1)\} \subset \mathbb{Z}_4^3,\] then
$\ind_3(A^{\mathbb{Z}_2^n}_3)$ is the largest cardinality of an $X$-free subset of $\mathbb{Z}_4^n$.
\end{theorem}

\begin{proof}
Recall that $\ind_3(A^{\mathbb{Z}_2^n}_3)$ is the largest cardinality of an $A_n$-star free subset of $(\mathbb{Z}_2^n)^2$, where $A_n = A^{\mathbb{Z}_2^n}_3$. So it suffices to
find a bijection $\psi$ from $(\mathbb{Z}_2^n)^2$ to $\mathbb{Z}_4^n$
such that $S\subseteq (\mathbb{Z}_2^n)^2$ is mapped to an $X$-free set if and only if $S$ is $A_n$-star free.

We define $\psi$ for $n=1$ and extend is entry-wise to a mapping from $(\mathbb{Z}_2^n)^2$ to $\mathbb{Z}_4^n$.
The definition for $n=1$ is as follows:
$\psi(0,0) = 0$, $\psi(0,1) = 1$, $\psi(1,0) = 2$ and $\psi(1,1) = 3$.

We need to show that if $(x_1,y_1), (x_2,y_2), (x_3,y_3) \in (\mathbb{Z}_2^n)^2$ is a $A_n$-star, then every coordinate in
$(\psi(x_1,y_1), \psi(x_2,y_2), \psi(x_3,y_3))$ belongs to $X$, and vice versa.
Since the map $\psi$ is defined coordinate-wise it suffices to check this for $n=1$. 
A triple $(x_1,y_1), (x_1,y_1+d), (x_1+d',y_1)$
is a (trivial or non-trivial) star in $A_1$ iff $x_1+(y_1+d)=(x_1+d')+y_1$, i.e., $d=d'$, and thus an $A_1$-star
is a triple of the form $(x_1,y_1), (x_1,y_1+d), (x_1+d,y_1)$.
If $d=0$ then obviously
$(\psi(x_1,y_1), \psi(x_1,y_1+d), \psi(x_1+d,y_1)) \in \{(0,0,0),(1,1,1),(2,2,2),(3,3,3)\} \subset X$. When $d=1$
there are four cases to check:
\begin{enumerate}

\item $x_1 = 0$ and $y_1=0$ then $(\psi(x_1,y_1), \psi(x_1,y_1+d), \psi(x_1+d,y_1))=(0,1,2) \in X$.

\item $x_1 = 0$ and $y_1=1$ then $(\psi(x_1,y_1), \psi(x_1,y_1+d), \psi(x_1+d,y_1))=(1,0,3) \in X$.

\item $x_1 = 1$ and $y_1=0$ then $(\psi(x_1,y_1), \psi(x_1,y_1+d), \psi(x_1+d,y_1))=(2,3,0) \in X$.

\item $x_1 = 1$ and $y_1=1$ then $(\psi(x_1,y_1), \psi(x_1,y_1+d), \psi(x_1+d,y_1))=(3,2,1) \in X$.

\end{enumerate}
On the other hand it is not hard to check that for each $(a,b,c) \in X$ the triplet $\psi^{-1}(a), \psi^{-1}(b), \psi^{-1}(c)$
is a star in $A_1$ or $a=b=c$. This proves the claim.
\end{proof}

Fix an integer $s\ge 2$ and let $HJ(n,s)$ denote the largest size of a $Y_s$-free subset of $[s]^n$,
where $Y_s$ is the following set of $s$-tuples:
$\{(1,\ldots,s)\}\cup\{(i,i,\ldots,i) | i=1,2,\ldots,s\}$.
The density Hales-Jewett theorem states that $HJ(n,s) = o(s^n)$
for every fixed $s$ \cite{furstenberg1991density, polymath2009new}.

Theorem~\ref{th:alt_char}, and the observation that the first three coordinates of
the $4$-tupples in $Y_4$ all bolong to $X$,
imply that $\ind_3(A^{\mathbb{Z}_2^n}_3) \le HJ(n,4)$.

The cap-set problem for $\mathbb{Z}_4^n$ also belongs to the same circle of problems.
It concerns the largest size of an arithmetic-triple-free set in $\mathbb{Z}_4^n$.
We mention in passing the recent breakthrough \cite{croot2016progression, ellenberg2017large} in this area
which showed that this size is at most $4^{(\gamma + o(1))\cdot n}$ with $\gamma \approx 0.926$.
Let $Z \subset \mathbb{Z}_4^3$ be the set of all ordered triplets
$(a,b,c)\in \mathbb{Z}^3_4$ satisfying $a + c = 2b$. The cap set problems concerns
exactly the largest possible cardinality of a $Z$-free subset of $\mathbb{Z}_4^n$.
Since $X \subset Z$ it follows that this size is bounded by $\ind_3(A^{\mathbb{Z}_2^n}_3)$.

The proof of Theorem~\ref{th:alt_char} extends verbatim to general $k\ge 3$.
It yields a subset $X \subset \mathbb{Z}_{2^{k-1}}^k$ such that
$\ind_k(A^{\mathbb{Z}_2^n}_k)$ is the largest cardinality of an $X$-free subset of
$\mathbb{Z}_{2^{k-1}}^n$.

By taking $X$ that includes all vectors
$(a,a,\ldots,a)\in \mathbb{Z}_{2^{k-1}}^k$ for $a \in \mathbb{Z}_{2^{k-1}}$ and the vector $(0,1,2,4,\ldots,2^{k-2})$
we can maintain the relation
between $\ind_k(A^{\mathbb{Z}_2^n}_k)$ and the density Hales-Jewett theorem for every $k$.

\section{Upper Bounds}

The results in this section are all restricted to case of three players. The extension to $k>3$ players
requires additional tools and ideas and is a good topic for future research.

\subsection{An Algorithmic Protocol for {\rm Exact-}$T$ over $\mathbb{Z}^d$}
\label{sec:a_protocol_for}

The aim of this section is to give the first algorithmic protocol for
\Exactn as well as other {\rm Exact-}$T$ functions. Our protocol is explicit, and
does not rely on a construction of a large set without a 3-term AP. We only appeal to the elementary fact that no sphere can contain three equally spaced colinear points.

The algorithm has two main steps. We first provide a very efficient protocol for {\rm Exact-}$T$  
over $\mathbb{Z}^d$, whose cost grows only logarithmically with $d$.

Let $f: ([m]^d)^3 \to \{0,1\}$ be defined via $f(x,y,z)=1$ if and only if  $x+y+z = T$, where
$T \in \mathbb{Z}^d$ is some fixed vector. We provide an explicit NOF protocol for $f$ whose cost is only $O(\log md)$. In words,  players try to compute the vector $x+2y+3z$ ``to the best of their knowledge" and then they compare notes. 

\begin{enumerate}

\item Player 1 computes $v_x = T-y-z + 2y + 3z$.

\item Player 2 computes $v_y = x + 2(T-x-z) + 3z$.

\item Player 3 computes $v_z = x + 2y + 3(T-x-y)$.

\item Player 1 writes $\|v_x\|_2^2$ on the blackboard.

\item Player 2 writes $1$ or $0$ on the blackboard depending on whether $\|v_y\|_2^2 = \|v_x\|_2^2$.

\item Player 3 writes $1$ or $0$ on the blackboard depending on whether $\|v_z\|_2^2 = \|v_x\|_2^2$.

\item The protocol outputs $1$ if the last two bits were both equal to $1$, and $0$ otherwise.

\end{enumerate}

The cost of the above protocol is essentially determined by the largest possible value of $\|v_x\|_2^2$ in step 4 which is at most $O(m^2d)$. Therefore, this cost does not exceed $O(\log md)$. We turn to prove correctness.
\begin{lemma}
The above protocol is correct.
\end{lemma}

\begin{proof}
First note that the protocol outputs $1$ if and only if $\|v_x\|_2^2 = \|v_y\|_2^2 = \|v_z\|_2^2$.
Also, $v_x + v_z = 2v_y$, so that this condition holds only if all three vectors are equal, in which case $T-x-y-z = 0$, as claimed. 
\end{proof}

\begin{remark}[More general protocols]
Several variations on the above theme suggest them selves. Fix
integers $a,b,c \in \mathbb{Z}$ and a $d \times d$ positive 
definite matrix $D$ with integer entries. The players compute
$a(T-y-z) + by + cz$, $ax + b(T-x-z) + cz$ and $ax + by + c(T-x-y)$, and rather
than comparing the values of $\|v\|_2$, they consider the values of $vDv^t$. We wonder
if these, or similar variations can together improve the complexity of the protocol.
\end{remark}

\subsection{Algorithmic Protocols for \Exactn and \ExactT over $\mathbb{Z}_m^d$}
\label{sec:reduction}

We seek algorithmically explicit  protocols for the exact-T problem over $\mathbb{Z}$
or equivalently over $\mathbb{Z}_n$.
That is the \Exactn problem with the function $f: [n]^3 \to \{0,1\}$ such that
$f(x,y,z) = 1$ if and only if $x+y+z=n$.  We can give an efficient 
protocol to this problem by reduction to the protocol in the previous section,
even though when applied directly to $\mathbb{Z}_n$ they give no improvement 
over the trivial protocol.
 
First fix a base $m$ and let $d = 1+\lceil \log_m n \rceil$. Consider the base-$m$ representation on the elements
of $[n]$. Given a representation $x \in \{0,1,\ldots,m-1\}^d$ of a number base $m$, for convenience we consider $x_1$ as the least significant digit. 
Note that all representations are of length $d$, if a number is small  its representation is padded with zeros.
The following protocol solves the \Exactn problem in these settings. Let $T$ be the base-$m$ representation of $n$.

\begin{enumerate}

\item Player $1$ computes the vector $C \in \{0,1,2\}^d$ defined as follows:
the $i$-th entry of $C$ is equal to $k \in \{0,1,2\}$ satisfying
$$
T_i + (k-1)m < y_i + z_i + C_{i-1} \le T_i + km,
$$
where addition is over $\mathbb{Z}$, and we define $C_0=0$.  

\item Denote by $C_x$ the carry vector computed by Player 1 in step 1. Player $2$ and $3$ compute corresponding vectors $C_y$ and $C_z$,
in a similar way.

\item Player $1$ writes $C=C_x$ on the board. 

\item Player $2$ and $3$, in turn, write $1$ on the board if and only if 
their vector $C_y$ ($C_z$) is equal to $C_x$.

\item If the last two bits written on the board are equal to $1$, continue.
Otherwise output $0$ and terminate.

\item All players compute (in private) the vector $T'_i = T_i + mC_i - C_{i-1}$,
for $i=1,\ldots,d$.

\item The players run a protocol for the exact-T problem over $\mathbb{Z}^d$ 
with $x,y,z$ and $T'$. 

\end{enumerate} 

The cost of the above protocol is $O(d+2)$ for steps 1-5, plus the cost of the protocol used in step 7. 
The cost is thus $O(d+\log md)$ if the players use the protocol from Section~\ref{sec:a_protocol_for} in the last step. 
We prove next that this protocol is correct.
\begin{lemma}
The above protocol is correct.
\end{lemma}

\begin{proof}
First assume $n=x+y+z$ over $\mathbb{Z}$. It is easy to verify the correctness 
of the protocol in this case, except maybe step 5. The correctness of step 5 follows from the 
following simple observation: 
assume $x_i+y_i+z_i+C_{i-1}= T_i + km$ (over $\mathbb{Z}$) for $k\in \{0,1,2\}$, then it must be that 
the sum of any pair of $x_i,y_i,z_i$ and $C_{i-1}$ is larger than $T_i + (k-1)m$ and at most
$T_i + km$. 

Now consider the case $T\ne x+y+z$. If the protocol rejects on step 5 then obviously
this is correct. If it does not reject then all players compute the same vector $T'$,
and $x+y+z = n$ over $\mathbb{Z}$ if and only if $x+y+z = T'$ over $\mathbb{Z}^d$.
The correctness now follows from the correctness of the protocol over $\mathbb{Z}^d$. 
\end{proof}

The above protocol for \Exactn is correct for any choice of base $m$.
To get an efficient protocol we optimize the choice of $m$. 
The running time of the protocol is $O(d+\log md) = O(d+\log m)$. Since $m^d = n$, 
we get that $\log n = d \log m$, and thus the optimal choice is roughly $m = 2^{\sqrt{\log n}}$
which gives a running time of $O(\sqrt{\log n})$.  

\begin{remark}[The group $\mathbb{Z}_m^d$]
The above protocol can also be adapted for $\mathbb{Z}_m^d$ (with addition modulo $m$). The idea is very similar,
the only difference is that in the first steps Player $1$ computes the vector 
$I_x \in \{0,1,2\}^n$ defined as follows:
the $i$-th entry of $I_x$ is equal to $k \in \{0,1,2\}$ satisfying
$$
T_i + (k-1)m < y_i + z_i \le T_i + km,
$$
where addition is over $\mathbb{Z}$. The other two players compute analogous vectors. 
\end{remark}

\subsection{A Protocol for \ExactT over $\mathbb{Z}_2^n$}

In this section we focus on the exact-$T$ problem for the abelian group $\mathbb{Z}_2^n$. In other words, we study
the permutation $f^{\mathbb{Z}_2^n}_{k,T}$. 
First we prove a lower bound on $\ind_3(f^{\mathbb{Z}_2^n}_{3,T})$, and then show that this lower bound
implies the existence of an efficient protocol for $f^{\mathbb{Z}_2^n}_{3,T}$. 
The complexity of $f^{\mathbb{Z}_2^n}_{k,T}$ is independent of $T$, so we can and will omit the subscript $T$ in this section. 
Throughout this subsection we let $A^{G}_{k} = A(f^{G}_{k})$.

First we prove that $A^{\mathbb{Z}_2^n}_{t}$-star freeness is preserved
under tensor product.\\Let $S \subset (\mathbb{Z}_2^n)^{k-1}$, denote by $S \otimes S$
the subset of $(\mathbb{Z}_2^{2n})^{k-1}$ comprised of all
vectors $(x_1,y_1, \ldots, x_{k-1},y_{k-1})$ such that $x_i, y_i \in S$
for $i=1,\ldots,k-1$.
\begin{lemma}
\label{lem:tensor_product}
If $S$ is $A^{\mathbb{Z}_2^n}_{k}$-star free then
$S \otimes S$ is $A^{\mathbb{Z}_2^{2n}}_{k}$-star free.
\end{lemma}

\begin{proof}
Let $A = A^{\mathbb{Z}_2^{2n}}_{k}$ and let
$$
(z_1,\ldots,z_{k-1}),(z_1+d,\ldots,z_{k-1}),\ldots,(z_1,\ldots,z_{k-1}+d)
$$
be an $A$-star in $S\times S$, where for each $1\le i \le k-1$, $z_i = (x_i, y_i)$ with $x_i, y_i \in S$.
Denote also $d = (d^1,d^2)$ where $d^1, d^2 \in \mathbb{Z}_2^n$. Then either
$$
(x_1,\ldots,x_{k-1}),(x_1+d^1,\ldots,x_{k-1}),\ldots,(x_1,\ldots,x_{k-1}+d^1)
$$
is an $A^{\mathbb{Z}_2^{n}}_{k}$-star in $S$, or
$$
(y_1,\ldots,y_{k-1}),(y_1+d^2,\ldots,y_{k-1}),\ldots,(y_1,\ldots,y_{k-1}+d^2)
$$
is an $A^{\mathbb{Z}_2^{n}}_{k}$-star in $S$, since either $d^1 \ne 0$ or $d^2 \ne 0$.
\end{proof}

It follows that if, for some fixed $m$, we can find a large $A^{\mathbb{Z}_2^{m}}_{k}$-star free subset $S$, then tensor powers of $S$ are large $A^{\mathbb{Z}_2^{n}}_{k}$-star free sets.
We show:
\begin{lemma}
\label{lem:n8}
$\ind_3(A^{\mathbb{Z}_2^{2}}_{3}) = \ind_3(n,n) = 8$.
\end{lemma}
Together with Lemma~\ref{lem:tensor_product} this yields:
\begin{corollary}
\label{cor:z2_ind}
For every integer $n \ge 2$, there holds $\ind_3(A^{\mathbb{Z}_2^{n}}_3) \ge 2^{3n/2}$.
\end{corollary}

\begin{proof}
Let $S$ be a star-free subset in $A^{\mathbb{Z}_2^{2}}_{3}$ of cardinality $8=4^{3/2}$ as
in Lemma~\ref{lem:n8}. The claim follows by taking the tensor powers of $S$ as in Lemma~\ref{lem:tensor_product}.
\end{proof}

\begin{proof}[Proof of Lemma~\ref{lem:n8}]
We denote the elements of $\mathbb{Z}_2^2$ as follows $(0,0)=0$, $(0,1)=1$, $(1,0)=2$ and $(1,1)=3$. The matrix associated with $A^{\mathbb{Z}_2^2}_{3}$ is:
\begin{center}
\begin{tabular}{ c c c c }
 \bf{0} & \bf{1} & 2 & 3\\
 1 & 0 & \bf{3} & \bf{2}\\
 \bf{2} & \bf{3} & 0 & 1\\
 3 & 2 & \bf{1} & \bf{0}
\end{tabular}
\end{center}
The $8$ entries in bold form a star-free set, so that $\ind_3(A^{\mathbb{Z}_2^{2}}_{3}) \ge 8$, and consequently
$\ind_3(4,4) \ge 8$. One can verify that in fact $\ind_3(4,4) = \ind_3(A^{\mathbb{Z}_2^{2}}_3) = 8$. To see this
first notice that if there is a star-free subset of cardinality $9$ then one of the values must appear three
times which already determines $10$ out of the $16$ entries. One can now rule out the existence of a size $9$
star-free subset by exhaustive search.
\end{proof}

It is interesting to determine $\ind_k(n,n)$ for some small values of $n$. For example:
\begin{itemize}
\item
Determine $\ind_3(8,8)$, in particular compute $\ind_3(A^{\mathbb{Z}_2^{3}}_3)$.
\item
Determine $\ind_k(4,4)$, in particular compute $\ind_k(A^{\mathbb{Z}_2^{2}}_k)$,
for $k>3$.
\end{itemize}

It is interesting to note that, while as shown, $\ind_3(4,4) = 8$,
there holds $\rho_3(4,4)=9$. The fact that $\rho_3(4,4)\le 9$ is easy to verify,
and the following example shows the equality:
\begin{center}
\begin{tabular}{ c c c c }
 \bf{1} & * & * & \bf{3}\\
 * & \bf{1} & * & \bf{4}\\
 * & * & \bf{1} & \bf{2}\\
 \bf{2} & \bf{3} & \bf{4} & *
\end{tabular}
\end{center}

Thus, continuing the discussion at the end of Section~\ref{sec:graph_theoretic_characterization}, $\rho_3(n,N)$
and $\alpha_3(n,N)$ need not be equal when $N<2n-1$.

The following theorem shows that for groups, $\alpha_k$ 
(the size of the largeset 1-monochromatic cylinder intersection) completely
characterizes $\chi_k$ (the minimum number of cylinder intersections
that partition the 1's). The proof
 is a simple generalization of Theorem 4.3 in \cite{CFL83}.
\begin{theorem}
\label{kappa_and_corners}
If $G$ is a group of order $n$, then
$$
\chr_k(f^G_{k}) \le O\left( \frac{kn^{k-1}\log n}{\ind_k(f^G_{k})} \right).
$$
\end{theorem}

\begin{proof}
The proof is in two steps:\\
{\bf Step I:} $A$-star freeness is preserved under translation, where $A=A^G_{k}$.
Indeed, let $S \subset G^{k-1}$ and let $\mathbf{a} = (a_1,\ldots,a_{k-1}) \in G^{k-1}$. If
$$(x_1,\ldots,x_{k-1}),(x_1+d,\ldots,x_{k-1}),\ldots,(x_1,\ldots,x_{k-1}+d)$$
is an $A$-star in $S+\mathbf{a}$, then
$$(x_1,\ldots,x_{k-1}) - \mathbf{a},(x_1+d,\ldots,x_{k-1}) - \mathbf{a},\ldots,(x_1,\ldots,x_{k-1}+d) - \mathbf{a}$$
is an $A$-star in $S$.\\
{\bf Step II:} Every $S \subset G^{k-1}$ has $O(\frac{kn^{k-1}\log n}{|S|})$
translates whose union covers all of $G^{k-1}$. This follows from the integrality gap for covering \cite{Lov75}, but
for completeness here is a proof.
Pick at random $t$ translates $\mathbf{a}_1, \ldots,\mathbf{a}_{t} \in [n]^{k-1}$ of $S$. The
probability that a given element $\mathbf{x} \in [n]^{k-1}$ is covered by a random translate of $S$ is exactly $\frac{|S|}{n^{k-1}}$. Therefore, and since the translates are picked independently uniformly at random, the expected number of uncovered elements of $G^{k-1}$ is
$$
n^{k-1} \cdot \left( 1-\frac{|S|}{n^{k-1}} \right)^t.
$$
Taking $t = O(\frac{kn^{k-1}\log n}{|S|})$ makes the expectation less than $1$, which proves the lemma.
\end{proof}

\begin{corollary}
\label{cor:bound_z_n}
There holds
$$
\chr_3(f^{\mathbb{Z}_2^m}_{3}) \le O\left( m\cdot 2^{m/2} \right).
$$
\end{corollary}

\begin{proof}
Follows from Theorem~\ref{kappa_and_corners} and Corollary~\ref{cor:z2_ind}.
\end{proof}

The bound in Corollary~\ref{cor:bound_z_n} is similar to the bound of Ada, Chattopadhyay, Fawzi and Nguyen \cite{ada2015nof} for the case $k=3$,
with slight improvement in the log factors. Ada et al. proved $\chr_3(f^{\mathbb{Z}_2^m}_{3}) \le O(m^{k+1} 2^{m/2^{k-2}})$,
by observing that this function is a composed function of the form $NOR \circ XOR$ and giving non trivial protocols
for such cases.

Note that the proof of Theorem~\ref{kappa_and_corners} yields a cover of $[n]^{k-1}$ by $A$-star free sets, but this is easily turned into a partition, since a subset of an $A$-star free set is also $A$-star free.
Therefore, any lower bound on $\ind_k(f^G_{k})$ can be translated into an upper bound on $\chr_k(f^G_{k})$ which in turn
implies an efficient (non-explicit) protocol for $f^G_{k}$ (By Theorem~\ref{D_chr_gf}). Another interesting consequence
of Theorem~\ref{kappa_and_corners} is that any lower bound on $\chr_k(f^G_{k})$ significantly larger than $\log n$
improves the known bounds for the size of a corner-free subset of $G$. This clearly boosts our interest in
the multiparty communication complexity of $f^G_{k}$.

We wonder whether there are analogs of Theorem~\ref{kappa_and_corners} for every permutation.
\begin{question}
How large can $\chr_k(A)\cdot \ind_k(A)/n^{k-1}$ be for an arbitrary permutation $A$?
\end{question}

\section{Lower Bounds}

\subsection{Nonconstructive Lower Bounds}

We first prove a nearly tight but nonconstructive lower bound on the communication complexity of random high-dimensional permutations.
\begin{theorem}
\label{th:random_permutation_cc}
For every integer $k \ge 3$, and for most $(k-1)$-dimensional permutations $f: [n]^k\to \{0,1\}$, $$\log \chr_k(f) \ge \Omega(\frac{\log n}{k}).$$
\end{theorem}

\begin{proof}
The lower bound on the number of high-dimensional permutations was recently improved by
Keevash \cite{keevash2018existence} who showed that there are at least $2^{\Omega(n^d\log n)}$ $d$-dimensional permutations. If we view a permutation as a map $[n]^k\to \{0,1\}$, this means at least $2^{\Omega(n^{k-1}\log n)}$ permutations. In the spirit of the proof of Lemma~{3.5} in \cite{BDPW07}, we now estimate the number of such permutation for which $\chr_k(f)$ is bounded.
Note that we cannot simply use the estimate from \cite{BDPW07} since it only works for functions $f: [n]^{k-1} \times [N] \to \{0,1\}$
with $N$ that is much smaller than $n$, roughly $N \le \sqrt{\frac{n}{k}}$.

Let $f: [n]^k\to \{0,1\}$ be a $(k-1)$-dimensional permutation, and let $\{C_1,\ldots,C_{\chi}\}$ be a partition of $f^{-1}(1)$ into $\chi=\chr_k(f)$ cylinder intersections.
For $i \in [k]$ define a function $A_i : [n]^{k-1} \to [\chi]$ as follows: For $a = (a_1,\ldots,a_{k-1}) \in [n]^{k-1}$, let $L = (a,i)$ be a line in $[n]^{k-1}$.
There is a unique $1$ entry in $L$ and this entry is in exactly one of the cylinder intersections $\{C_1,\ldots,C_{\chi}\}$, say $C_j$. In this case we define $A_i(a_1,\ldots,a_{k-1})=j$.

As seen in the proof of Theorem~\ref{D_chr_gf}, it is possible to recover $f$ from knowledge of
the functions $A_1,\ldots,A_k$. Namely, $f(x_1,\ldots,x_k) = 1$
if and only if all the values $A_i(x_1,\ldots,x_{i-1},x_{i+1},\ldots,x_{k-1})$ for $i=1,\ldots ,k$ are equal.
But for every $i \in [k]$ there are $\chi^{n^{k-1}}$ possible functions $A_i : [n]^{k-1} \to [\chi]$.
Thus, the number of $(k-1)$-dimensional permutations $f: [n]^k\to \{0,1\}$
with $\chr_k(f) \le \chi$ is at most $(\chi^{n^{k-1}})^k = 2^{kn^{k-1}\cdot\log \chi}$. Combining this with Keevash's lower bound we get that
\[
\log \chi \ge \Omega(\frac{\log n}{k}).
\]
for most $(k-1)$-dimensional permutations.
\end{proof}

A simple corollary of Theorem~\ref{th:random_permutation_cc},  and Theorem~\ref{D_chr_gf} is:
\begin{corollary}
For every integer $k \ge 2$, almost all
$(k-1)$-dimensional permutations $f: [n]^k\to \{0,1\}$ satisfy
\[
D_k(f) \ge \Omega(\frac{\log n}{k}).
\]
\end{corollary}

Theorem~\ref{th:random_permutation_cc} proves the lower bound $\chr_k(f) \ge 2^{\Omega(\frac{\log n}{k})}$ for a random
permutation $f: [n]^k \to \{0,1\}$.
It is interesting to find out how this extends for a random linjection $f: [n]^{k-1}\times [N] \to \{0,1\}$ with $n < N$. It is also interesting to see whether the dependency
on $k$ can be removed.

Finally we turn to the case $k=3$. The number of 2-dimensional permutations (aka Latin squares) is known to be $((1+o(1))\frac{n}{e^2})^{n^2}$ (see \cite{van2001course}).
It follows that for most $2$-dimensional permutations $f$ there holds $\log \chr_3(f) \ge \frac{1}{3}\log n-\Theta(1)$.

\subsection{Lower Bounds for $\chi_k(n,N)$}

We  prove an upper bound on $\ind_k(n,N)$, using its graph theoretic interpretation from
Section~\ref{sec:graph_theoretic_characterization}. This naturally entails a corresponding
lower bound on $\chi_k$. We start with the case $k=3$.

\begin{theorem}
\label{1}
Let $A: [n] \times [n] \to [N]$ be a linjection, where $N \le n\cdot 2^{c\log^* (n)}$. Then,
$$
\ind_3(A) \le O\left(\frac{n^2}{2^{c\log^* (n)}}\right).
$$
Here $c > 0$ is an absolute constant.
\end{theorem}

The proof of Theorem~\ref{1} is an adaptation of Solymosi's \cite{Solymosi2003}
simplification of
Ajtai and Szemer\'edi's \cite{ajtai1974sets} Corners Theorem. 
We will use the improved version of the
triangle removal lemma \cite{ruzsa1978triple} due to Fox \cite{fox2010}:
\begin{lemma}[Triangle removal lemma]
\label{2}
For every $\epsilon >0$ there is a $\delta >0$ such that every $n$-vertex graph with at most $\delta n^3$ triangles can be made
triangle-free by removing $\epsilon n^2$ edges. Specifically $\delta^{-1}$ can be taken as a tower of twos of height $405 \log \epsilon^{-1}$.
\end{lemma}

\begin{proof}[of Theorem~\ref{1}]
Let $G=G(A)$, $V = V(G)$. Notice that $|V| = 2n+N$.
Let $S \subset [n]^2$ be an $A$-star free subset of size $\ind_3(A)$.
As in the proof of Theorem~\ref{equiv-rho-kappa}
we let $T = \{<x,y,A(x,y)>|(x,y) \in S\}$ be the family of triangles in $G$ that corresponds to $S$.
Let $F$ be that subgraph of $G$ whose edge set is the union of all triangles in $T$. This graph contains the $|S|$
edge-disjoint triangles in $T$, and no additional triangles.

Thus, if we denote $\delta = |S|/|V|^3$ and $\epsilon = |S|/|V|^2$, then $F$ contains exactly $\delta |V|^3$ triangles
and it cannot be made triangle free by removing fewer than $\epsilon |V|^2$ edges.
Lemma~\ref{2} yields $\log^* (\delta^{-1}) \le 405 \log (\epsilon^{-1})$, and
since $\delta < \frac{n^2}{(2n+N)^3} < \frac{1}{N}$ we conclude that
$$
\epsilon \le 2^{\frac{-1}{405}\log^* (N)}.
$$
But $|S| = \epsilon |V|^2 \le 9 \epsilon N^2$, so that for $N \le 2^{c \log^* (n)} n$, with $c = (3\cdot 405)^{-1}$, there holds
$$
|S| \le O\left(\frac{n^2}{2^{c\log^* (n)}}\right).
$$
\end{proof}

We now proceed to the case of general $k$.

\label{sec:the_case_of_a_general_k}

\begin{theorem}
\label{1-general-k}
For every natural numbers $k\ge 3$ ,$n$ and $N$ it holds that
$$
\ind_k(n,N) \le O\left(\frac{kn^{k-2}N}{\log^* (n)}\right).
$$
\end{theorem}

To this end we need the hypergraph removal lemma.
\begin{theorem}[\cite{gowers2007hypergraph, nagle2006counting, rodl2004regularity, rodl2006applications, tao2006variant}]
\label{2-general-k}
Let $k$ be a positive integer. For every $\epsilon > 0$ there exists $\delta > 0$
with the following property. Let $H$ be a $k$-partite $(k-1)$-uniform hypergraph with parts $X_1,\ldots, X_k$ and at most
$\delta \Pi_{i=1}^k |X_i|$ cliques. There exists for each $i$, a subset $R_i\subseteq \Pi_{j\ne i} X_j$ of at most
$\epsilon \Pi_{j\ne i} |X_j|$ hyperedges of $H$ so that the hypergraph $H\setminus\cup R_i$ is clique-free.
Specifically one can take $\delta^{-1}$ to be a tower
of twos of hight $O(\epsilon^{-1})$.
\end{theorem}

\begin{proof}[of Theorem~\ref{1-general-k}]
By Theorem~\ref{equiv-rho-kappa-general-k} $\ind_k(n,N) \le \rho_k(n,N)$, so it suffices to prove that
$$
\rho_k(n,N) \le O\left(\frac{kn^{k-2}N}{\log^* (n)}\right).
$$
By definition of $\rho_k$, there is a $k$-partite $(k-1)$-uniform hypergraph $H$ with vertex sets $X_1=\cdots=X_{k-1}=[n]$
and $X_k = [N]$, containing exactly $\rho_k(n,N)$ disjoint $k$-cliques and no additional cliques.
Consequently, at least
$\rho_k(n,N)$ hyperedges must be removed to make $H$ clique-free, whence
$$
\rho_k(n,N) \le \epsilon kn^{k-2}N.
$$
But $\delta^{-1} \ge n$, so that
$\epsilon = O(\frac{1}{\log^* \delta^{-1}}) = O(\frac{1}{\log^* n})$. The claim follows.
\end{proof}

\subsection{A Lower Bound on $\chr_3(n,N)$}
\label{sec:a_lower_bounds_on_chr_3_n_N}

In this section we prove a better lower bound for the case $k=3$.

\label{sec:lower_bound_chr_3}

\begin{theorem}
\label{lb-chi_3}
$\chr_3(n,n) \ge \log \log n - O(\log\log\log n)$.
\end{theorem}

This is clearly the case $N=n$ of the following lemma.

\begin{lemma}
\label{lb-1}
Let $L = \chr_3(n,N)$ for some integers $N\ge n$, then
\[\log n<(2^{L+1}-1)\cdot\log(4NL/n).\]
\end{lemma}

\begin{proof}[of Lemma \ref{lb-1}]
Let $A: [n] \times [n] \to [N]$ be a linjection with $\chr_3(A) = L$. This means that $A$'s entries can be $L$-colored
so that every color class is $A$-star free. We pick $v_1 \in [N]$, the most frequent value that appears in $A$, and then $c_1\in [L]$, the most abundant color among $A$'s $v_1$-entries. Clearly, $|S_1| \ge n^2/(N L)$,
where $S_1$ is the set of $c_1$-colored $v_1$-entries in $A$. As usual we denote the closure of $S_1$ by $\bar{S}_1$, and note
that since $S_1$ meets every row and column in $A$ at most once, there is a combinatorial rectangle
$R_1 \subseteq \bar{S}_1\setminus S_1$ of sides $\frac{|S_1|}2\times\frac{|S_1|}2$. Clearly the color $c_1$
is missing from $R_1$.

Now we recurse:
Let $v_i \in [N]$ be the most frequent value that appears in $R_{i-1}$, and $c_i$ the most abundant color among these entries. Let $S_i$ be the set of $c_i$-colored $v_i$-entries in $R_{i-1}$. Finally,
$R_i \subseteq \bar{S}_i\setminus S_i$ is a combinatorial rectangle of sides $\frac{|S_i|}2\times\frac{|S_i|}2$
that misses colors $c_1, \ldots, c_{i-1}, c_i$. It follows that for all $i \ge 1$ there holds

\[|S_{i+1}|\ge\frac{|S_i|^2}{4NL},\]

which yields by induction that
\[ |S_i| \ge \frac{n^{2^{i}}}{4^{2^{i-1}-1}(NL)^{2^{i}-1}}.\]

But since we eliminate one letter each time, this reduction process can last at most $L$ steps,
namely $|S_{L+1}| \le 1$, whence

$$
1> \frac{n^{2^{L+1}}}{(4NL)^{2^{L+1}-1}} = n\cdot\left(\frac{n}{4NL}\right)^{2^{L+1}-1}
$$
as claimed.
\end{proof}

Another simple corollary of Lemma~\ref{lb-1} is due to Meshulam and is reproduced in \cite{alon2012nearly}.

\begin{corollary}
\label{cor:chr_N}
If $\chr_3(n,N) \le L$ for some integers $N\ge n$, then $N \ge \frac{1}{4L}\cdot n^{1+1/(2^L-1)}$.
\end{corollary}

\medskip

\noindent {\bf A note on the case $k > 3$.}

As we have just seen $\chr_3(A) \ge \Omega(\log \log n)$ for every $2$-dimensional
permutation $A$. It is conceivable that a similar bound holds for higher dimensions as well.
This was previously conjectured in~\cite{beigel2006multiparty} for the Exact-$T$ problem.
If we try to adapt the proof of Lemma~\ref{lb-1} to higher $k$, exactly one difficulty arises
which we formulate as a question.

\begin{question}
\label{q_scs_k}
Let $S \subseteq [n]^{k}$ be a set of cardinality $m$ that meets every line at most once. Determine, or estimate $\phi_k(n, m)$, the least possible cardinality $|\bar{S}|$ of its closure. We use the shorthand $\phi_k(m)$ when appropriate.
\end{question}

For $k=2$ the answer is easy: $\phi_2(m)=m^2$, since $|\bar{S}| = |S|^2$.
But for $k > 2$ the problem becomes very hard and no lower bound is known.
In fact, for $k \ge 3$, and for large enough $m$ there holds $\phi_k(m)=m$. In other words,
unlike the case $k=2$ it may happen that $\bar{S}=S$ for large $S$. For example, as shown in \cite{CFL83}, $\phi_3(m)=m$
when $m=n^2/2^{\Omega(\sqrt{\log n})}$, whereas it is shown in \cite{shkredov2009two} that $\phi_3(m) > m$
when $m \ge n^2/(\log\log n)^{\frac{1}{22}}$.
For $k > 3$ the situation is even worse, and all we have are the very weak lower bounds from
Section~\ref{sec:the_case_of_a_general_k}. Namely, it follows from Theorem~\ref{1-general-k} that $\phi_k(m)$
must be larger than $m$ when $m \ge \Omega\left(\frac{kn^{k-1}}{\log^* (n)}\right)$, and that is all we know.

It should be clear that proving any non-trivial bounds on $\phi_k(m)$ is a very interesting
challenge. We raise the following conjecture in an attempt of improving the lower bounds on $\chr_3(n,n)$.
\begin{conjecture}
There are constants $c_1, c_2 > 0$ such that if $S \subseteq [n]^{3}$ meets
every line at most once, and if $|S| \ge n^2/(\log\log n)^{c_1}$, then $|\bar{S}| \ge n^3/(\log\log n)^{c_2}$.
\end{conjecture}

\section{Conclusion and Open Problems}

This paper raises numerous open problems. Below we collect some of the major ones and explain
some implications that would follow from progress on these questions.

\begin{question}
Improve the lower bound $\chr_3(n,n) \ge \Omega(\log \log n)$.
\end{question}

Implications:

\begin{itemize}

\item Any lower bound $\chr_3(n,n) \ge \omega(\log \log n)$ yields an improvement to the best known bound
on the number of colors required to color the $n \times n$ grid with no monochromatic
equilateral right triangles. This subject goes back to Ajtai and Szemer\'{e}di's corners theorem \cite{ajtai1974sets} and its implications in additive combinatorics due to Solymosi \cite{Solymosi2003}.

\item A lower bound $\chr_3(n,n) \ge \omega(\log n)$ would improve the best known gap between
randomized and deterministic communication complexity in the 3-players NOF model.

\item  A lower bound $\chr_3(n,n) \ge \Omega(\log n \cdot \log\log n)$ will improve the best
known upper bound on the size of corner-free subsets of $G^2$ for any abelian group $G$.

\item  A lower bound $\chr_3(n,n) \ge \Omega(\log^2 n)$ will improve the best
bounds on the size of a subset of $\mathbb{Z}_n$ with no three-term arithmetic progression. This is a classic problem that goes
back at least to the 1950's \cite{roth1953certain}.

\end{itemize}

\begin{question}
Improve the upper bound $\chr_3(n,n) \le 2^{O(\sqrt{ \log n})}$.
\end{question}

Implications:

\begin{itemize}
\item
The construction of denser \rs\ graphs than
currently known. Namely, $n$-vertex graphs which are the disjoint union of $n$ induced matchings, all of the same size $r$. This,
in turn, reflects on the many applications of these.

\item
That would improve our understanding regarding the limits of the triangle removal lemma.
Note that the current gaps between the bound in this lemma are huge.
\end{itemize}

\begin{question}
Improve the bounds on $\chr_k(n,n)$ for $k > 3$.
\end{question}

\begin{question}
Improve the bounds on $\ind_k(n,n)$ for $k > 3$.
\end{question}

That would improve our state of knownledge regarding the bounds for the hypergraph removal lemma.

It is also interesting to determine $\ind_k(n,n)$ for some small values of $n$. For example:
\begin{itemize}
\item
Determine $\ind_3(8,8)$, in particular compute $\ind_3(A^{\mathbb{Z}_2^{3}}_3)$.
\item
Determine $\ind_k(4,4)$, in particular compute $\ind_k(A^{\mathbb{Z}_2^{2}}_k)$,
for $k>3$.
\end{itemize}

\begin{question}
What is the relationship between $\rho_k(n,N)$ and $\ind_k(n,N)$ in the whole range $n\le N \le (k-1)(n-1)$?
\end{question}

\bibliographystyle{plain}
\bibliography{hdp}

\end{document}